\documentclass[12pt]{iopart}
\usepackage{amsthm}
\usepackage{graphicx}

\font\bb=msbm10 at 12pt
\newcommand{\mathbb}[1]{\hbox{\bb #1}}
\newtheorem{claim}{Claim}[section]
\newtheorem{theorem}[claim]{Theorem}

\newtheorem{definition}[claim]{Definition}
\newtheorem{corollary}[claim]{Corollary}
\begin{document}

\title[On the Effective Size of a Non-Weyl Graph]{On the Effective Size of a
Non-Weyl Graph}
\author{Ji\v{r}\'{\i} Lipovsk\'{y}}

\address{Department of Physics, Faculty of Science, University of Hradec
Kr\'{a}lov\'{e}, Rokitansk\'{e}ho 62, 500\,03 Hradec Kr\'{a}lov\'{e},
Czechia} \ead{jiri.lipovsky@uhk.cz} 
\vspace{10pt} 
\begin{indented}
\item[]June 2016
\end{indented}

\begin{abstract}
We show how to find the coefficient by the leading term of the resonance
asymptotics using the method of pseudo orbit expansion for quantum graphs
which do not obey the Weyl asymptotics. For a non-Weyl graph we develop a
method how to reduce the number of edges of a corresponding directed graph.
Through this method we prove bounds on the above coefficient depending on
the structure of the graph for graphs with the same lengths of internal
edges. We explicitly find the positions of the resolvent resonances.
\end{abstract}

\vspace{2pc} \noindent\textit{Keywords}: quantum graphs, resonances, Weyl
asymptotics %
%\maketitle

\section{Introduction}

Quantum graphs have been intensively studied mainly in the last thirty
years. There is extensive bibliography we refer to, e.g. the book \cite{BK},
the proceedings \cite{AGA} and references therein. One of the models studied
are quantum graphs with attached semiinfinite leads. For this model, the
notion of resonances can be defined; there are two main definitions:
resolvent resonances (singularities of the resolvent) and scattering
resonances (singularities of the scattering matrix). For the study of
resonances in quantum graphs we refer e.g. to \cite{ESr, Ex2, Ex3, KSc, BSS,
EL1, EL2}.

The problem of finding resonance asymptotics in quantum graphs has been
addressed in \cite{DP, DEL, EL3}. A surprising observation by Davies and
Pushnitski \cite{DP} shows that a graph has in some cases fewer resonances
than one would expect from Weyl asymptotics. Criteria, which can distinguish
these non-Weyl graphs from the graphs with regular Weyl asymptotics, have
been presented in \cite{DP} (graphs with standard coupling), and in \cite%
{DEL} (graphs with general coupling). Although distinguishing between these
two cases is quite easy (it depends on the vertex properties of the graph),
finding the leading term of asymptotics, which is closely related to the
``effective size'' of the graph, is more
difficult since it uses the structure of the whole graph.

In the paper we express the first term of non-Weyl asymptotics using the
method of pseudo orbits and find bounds on the effective size that depend
on the structure of the graph. The paper is divided into the following eight
sections: in the second section we introduce the model itself, in the third
section we state known theorems on the asymptotics of the resonances, in the
fourth section we develop the method of pseudo orbit expansion for the
resonance condition, in the fifth section we state how the effective size
can be found for an equilateral graph (a graph with the same lengths of the
internal edges), in the sixth section we develop a method that allows us to
delete some of the edges of a non-Weyl graph, in the seventh section we
state the main theorems on the bounds on the effective size and position of
the resonances for equilateral graphs, and finally the eighth section
illustrates found results and developed method in three examples.

\section{Preliminaries}

First, we describe the main notions. We consider a metric graph $\Gamma $
that consists of the set of vertices $\mathcal{V}$, the set of $N$ finite
internal edges $\mathcal{E}_{\mathrm{i}}$ of lengths $\ell _{j}$ (these
edges can be parametrized by intervals $(0,\ell _{j})$), and the set of $M$
semiinfinite external edges $\mathcal{E}_{\mathrm{e}}$ (parametrized by
intervals $(0,\infty )$). The graph is equipped by the second order
differential operator $H$ that acts as $-\mathrm{d}^{2}/\mathrm{d}x^{2}$ on
internal and external edges. The domain of this operator consists of
functions with edge components in Sobolev space $W^{2,2}(e_{j}),$ and at the
same time satisfying the coupling conditions at the vertices 
\begin{equation}
(U_{j}-I)\Psi _{j}+i(U_{j}+I)\Psi _{j}'=0\,,  \label{eq-coupling1}
\end{equation}%
where $U_{j}$ is a $d_{j}\times d_{j}$ unitary matrix ($d_{j}$ is the degree
of the vertex), $I$ is a $d_{j}\times d_{j}$ identity matrix, $\Psi _{j}$ is
the vector of limits of functional values to the vertex from attached edges
and, similarly, $\Psi _{j}'$ is the vector of limits of outgoing
derivatives. This general form of the coupling was independently introduced
by Kostrykin and Schrader \cite{KS2} and Harmer \cite{Ha}.

There is a trick shown in \cite{EL2, Ku3} how to describe the coupling on
the whole graph by one big $(2N+M)\times (2N+M)$ unitary matrix $U$ that is
similar to a block diagonal matrix with blocks $U_{j}$. (Square matrices $A$ 
and $B$ are similar if there exists a regular square matrix~$P$ 
such that $A = P B P^{-1}$.) This matrix encodes
not only the coupling but also the topology of the graph. 
\begin{equation}
(U-I)\Psi +i(U+I)\Psi'=0\,.  \label{eq-coupling2}
\end{equation}%
In the previous equation $I$ is $(2N+M)\times (2N+M)$ identity matrix, and
the vectors $\Psi $ and $\Psi'$ consist of entries $\Psi _{j}$ and 
$\Psi _{j}'$. This coupling condition corresponds to a graph where
all the vertices are joined into one vertex.

One can describe the effective coupling on the compact part of the graph
using energy-dependent $2N\times 2N$ coupling matrix $\tilde{U}(k)$, where $k
$ is the square root of energy. There is a straightforward way \cite{EL2}
how to construct this effective coupling matrix by a standard method by
Schur, etc. 
\begin{equation}
\tilde{U}(k)=U_{1}-(1-k)U_{2}[(1-k)U_{4}-(1+k)I]^{-1}U_{3}\,,
\label{eq-utilde}
\end{equation}%
where the matrices $U_{1},\dots ,U_{4}$ are blocks of the matrix $U=\left( 
\begin{array}{cc}
U_{1} & U_{2} \\ 
U_{3} & U_{4}%
\end{array}%
\right) $; $U_{1}$ corresponds to coupling between internal edges, $U_{4}$
corresponds to coupling between external edges, and $U_{2}$ and $U_{3}$
correspond to the mixed coupling. The coupling condition has a similar form
to (\ref{eq-coupling2}) with $U$ replaced by $\tilde{U}(k)$ and $I$ meaning
now $2N\times 2N$ identity matrix.

\section{Asymptotics of the number of resonances}

\label{sec-asym} We are interested in asymptotical behavior of the number of
resolvent resonances of the system. The \emph{resolvent resonance} is
usually defined as the pole of the meromorphic continuation of the resolvent 
$(H-\lambda \mathrm{id})^{-1}$ into the second sheet. The resolvent
resonances can be obtained by the method of complex scaling (more on that in 
\cite{EL1}). We use a simpler but equivalent definition.

\begin{definition}
We call $\lambda = k^2$ a resolvent resonance if there exists a generalized
eigenfunction $f\in L^2_{\mathrm{loc}}(\Gamma)$, $f\not\equiv 0$ satisfying 
$-f''(x) = k^2 f(x)$ on all edges of the graph and fulfilling the
coupling conditions (\ref{eq-coupling1}) that on all external edges behaves
as $c_j \,\mathrm{e}^{ikx}$.
\end{definition}

This definition is equivalent to the previously mentioned one because these
generalized eigenfunctions (which, of course, are not square integrable),
become after complex scaling square integrable. It was proven in \cite{EL1}
that the set of resolvent resonances is equal to the union of the set of
scattering resonances, and the eigenvalues with the eigenfunctions supported
only on the internal part of the graph.

We define the counting function $N(R)$ that counts the number of all
resolvent resonances including their multiplicities contained in the circle
of radius $R$ in the $k$-plane. One should note that by this method we count
twice more resonances than in the energy plane. This is clear for the case
of a compact graph because $k$ and $-k$ correspond to the same eigenvalue.

Asymptotics of the resonances in most of the quantum graphs obeys Weyl's
law~-- equation~(\ref{eq-weyl}) with $W$ equal to the sum of the lengths of
the internal edges, i.e. $\mathrm{vol\,}\Gamma $. However, there exist such
quantum graphs for which the leading term of the asymptotics is smaller than
expected; we call these graphs \emph{non-Weyl}. The problem was studied for
graphs with standard coupling (in literature also called Kirchhoff, free or
Neumann coupling) in \cite{DP}, and for graphs with general coupling in 
\cite{DEL}. We define the standard coupling and state main results presented
in these papers.

\begin{definition}
For the standard coupling functional values are continuous at each vertex
and the sum of outgoing derivatives is equal to zero, i.e. for the vertex $v$
with degree $d$ we have 
\begin{eqnarray*}
f_i (v) = f_j(v) \equiv f(v)\,, \quad \forall i, j \in 1, \dots, d, \\
\sum_{j = 1}^d f_j'(v) = 0\,.
\end{eqnarray*}
\end{definition}

\begin{theorem}
(Davies and Pushnitski)\newline
Let us assume a graph $\Gamma $ with standard coupling at all vertices and
the sum of the lengths of all internal edges equal to $\mathrm{vol\,}\Gamma $. 
Then the number of resolvent resonances has the following asymptotics 
$$
N(R)=\frac{2}{\pi }WR+\mathcal{O}(1)\,,\quad \mathrm{as\ }R\rightarrow
\infty \,, 
$$
where $0\leq W\leq \mathrm{vol\,}\Gamma $. One has $W<\mathrm{vol\,}\Gamma $
iff there exists at least one balanced vertex, i.e. a vertex that connects
the same number of internal and external edges.
\end{theorem}

\begin{theorem}
(Davies, Exner and Lipovsk\'{y})\newline
Let us assume a graph $\Gamma$ with general coupling (\ref{eq-coupling2})
and with the sum of the internal edges equal to $\mathrm{vol\,}\Gamma$. The
number of resolvent resonances has the asymptotics 
\begin{equation}
N(R) = \frac{2}{\pi} WR + \mathcal{O}(1)\,,\quad \mathrm{as\ }R \to \infty\,,
\label{eq-weyl}
\end{equation}
where $0\leq W \leq \mathrm{vol\,}\Gamma$. One has $W < \mathrm{vol\,}\Gamma$
(the graph is non-Weyl) iff the effective coupling matrix $\tilde U(k)$ has
at least one eigenvalue equal to either $\frac{1+k}{1-k}$ or $\frac{1-k}{1+k}
$ for all $k$.
\end{theorem}

While it is quite easy to find whether a graph is Weyl or non-Weyl (it
depends only on the vertex properties of the graph), determining the
coefficient $W$ (the effective size of the graph) is more complicated,
because it depends on the properties of the whole graph. It is illustrated
e.g. in Theorem 7.3 in \cite{DEL}.

\section{Pseudo orbit expansion for the resonance condition}

\label{sec-pseudo} There is a theory developed showing how to find the
spectrum of a compact quantum graph by the pseudo orbit expansion (see e.g. 
\cite{BHJ, KoS2}). In this section we adjust this method to find the
resonance condition. The idea is to find the effective vertex-scattering
matrix that acts only on the compact part of the graph. The
vertex-scattering matrix maps the vector of amplitudes of the incoming waves
to the vertex into the vector of amplitudes of the outgoing waves. We show
that in the case of standard coupling this matrix is not energy dependent
and has a nicely arranged form.

\begin{definition}
Let us assume a vertex $v$ of the graph $\Gamma$. Let there be $n$ internal
edges emanating from $v$, all parametrized by $(0, \ell_j)$ with $x=0$
corresponding to $v$, and $m$ halflines. Let the solution of Schr\"odinger
equation on these internal edges be $f_j(x) = \alpha_j^{\mathrm{in}} 
\mathrm{e}^{-ikx}+\alpha_j^{\mathrm{out}} \mathrm{e}^{ikx}$, $j = 1, \dots, n$ and
on the external edges $g_s(x) = \beta_s\mathrm{e}^{ikx}$, $s = 1, \dots, m$.
Then the $n\times n$ effective vertex-scattering matrix $\tilde \sigma^{(v)}$
is given by the relation $\vec{\alpha}_v^{out} = \tilde \sigma^{(v)} 
\vec{\alpha}_v^{in}$, where $\vec{\alpha}_v^{in}$ and $\vec{\alpha}_v^{out}$ are
vectors with entries $\alpha_j^{\mathrm{in}}$ and $\alpha_j^{\mathrm{out}}$,
respectively.
\end{definition}

For the next theorem we drop the superscript or subscript $v$ denoting the
vertex in coupling and vertex-scattering matrices.

\begin{theorem}
\label{thm-evsm}(general form of the effective vertex-scattering matrix)\newline
Let us consider a non-compact graph $\Gamma$ with the coupling at the vertex 
$v$ given by (\ref{eq-coupling1}) and the matrix $U$. Let the vertex $v$
connect $n$ internal edges and $m$ external edges. Let the matrix $I$ be 
$(n+m)\times (n+m)$ identity matrix and by $I_n$ we denote the $n\times n$
identity matrix and by $I_m$ the $m\times m$ identity matrix. Then the
effective vertex-scattering matrix is given by $\tilde\sigma (k) = -
[(1-k)\tilde U(k)-(1+k)I_n]^{-1}[(1+k)\tilde U(k)-(1-k)I_n]$. The inverse
relation is $\tilde U(k) = [(1+k)\tilde\sigma(k) + (1-k)I_n][(1-k)\tilde
\sigma(k)+(1+k)I_n]^{-1}$.
\end{theorem}

\begin{proof}
Let the wavefunction components on the internal edges emanating from the 
vertex~$v$ be $f_j(x) = \alpha_j^{\mathrm{in}}\mathrm{e}^{-ikx}+\alpha_j^{\mathrm{out}}\mathrm{e}^{ikx}$ 
and on the external edges $g_j(x) = \beta_j \mathrm{e}^{ikx}$. The vector of 
functional values therefore is 
$\Psi = \left(\begin{array}{c}\vec\alpha^{\mathrm{in}}+\vec\alpha^{\mathrm{out}}\\ \vec\beta\end{array}\right)$ 
and the vector of outgoing derivatives is 
$\Psi' = ik \left(\begin{array}{c}-\vec\alpha^{\mathrm{in}}+\vec\alpha^{\mathrm{out}}\\ \vec\beta\end{array}\right)$. 
The coupling condition (\ref{eq-coupling1}) gives
$$
  (U-I)\left(\begin{array}{c}\vec\alpha^{\mathrm{in}}+\vec\alpha^{\mathrm{out}}\\ \vec\beta\end{array}\right) + i i k (U+I)\left(\begin{array}{c}-\vec\alpha^{\mathrm{in}}+\vec\alpha^{\mathrm{out}}\\ \vec\beta\end{array}\right) = 0\,.
$$
 Hence we have the set of equations
\begin{eqnarray*}
 \hspace{-5mm}[U_1-I_n-k(U_1+I_n)]\vec\alpha^{\mathrm{out}}+ [(U_1- I_n)+ k(U_1+ I_n)]\vec\alpha^{\mathrm{in}}+(1-k)U_2 \vec\beta = 0\,,\\
 \hspace{-5mm}(1-k) U_3 \vec\alpha^{\mathrm{out}}+(1+k)U_3 \vec\alpha^{\mathrm{in}}+[(U_4-I_m)-k(U_4+I_m)]\vec\beta = 0\,.
\end{eqnarray*}
Expressing $\vec \beta$ from the second equation and substituting it in the first one one has
\begin{eqnarray*}
  \hspace{-1cm}\{(1-k)U_1-(1+k)I_n-(1-k)U_2[(1-k)U_4-(1+k)I_m]^{-1}(1-k)U_3\}\vec\alpha^{\mathrm{out}}+ 
  \\
  \hspace{-1cm}+\{(1+k)U_1-(1-k)I_n-(1-k)U_2[(1-k)U_4-(1+k)I_m](1+k)U_3\}\vec\alpha^{\mathrm{in}} = 0\,,
\end{eqnarray*}
from which using (\ref{eq-utilde}) the claim follows. Expressing the inverse relation is straightforward.
\end{proof}

Using the previous theorem one can straightforwardly compute the effective
vertex-scattering matrix for standard conditions.

\begin{corollary}
\label{cor-stan} Let $v$ be a vertex with $n$ internal and $m$ external
edges with standard coupling conditions (i.e. $U = \frac{2}{n+m}%
J_{n+m}-I_{n+m}$, where $J_n$ denotes $n\times n$ matrix with all entries
equal to one). Then the effective vertex-scattering matrix is $\tilde
\sigma(k) = \frac{2}{n+m}J_n - I_n$, in particular, for a balanced vertex we
have $\tilde \sigma(k) = \frac{1}{n}J_n - I_n$.
\end{corollary}

\begin{proof}
There are two ways how to prove this corollary. The first is to compute the 
effective vertex-scattering matrix from the definition in a similar way to the 
proof of Theorem~\ref{thm-evsm}. The latter, to some extent longer, but 
straightforward, uses the Theorem~\ref{thm-evsm} itself. Using the formula 
$(aJ_n+bI_n)^{-1} = \frac{1}{b}\left(-\frac{a}{an+b}J_n+I_n\right)$ we compute 
the effective coupling matrix and obtain $\tilde U(k) = \frac{2}{k m +n}J_n-I_n$. 
Substituting it in the formula in Theorem \ref{thm-evsm} and using the above 
expression of the inverse matrix, and minding the fact that 
$J_n \cdot J_n = n J_n$ one obtains the result. 
\end{proof}

Having defined the effective vertex-scattering matrix, we can proceed to the
lines of the method shown in \cite{BHJ}. The only difference is that the
vertex-scattering matrix $\sigma ^{(v)}$ introduced in \cite{BHJ} is
replaced by the effective vertex-scattering matrix $\tilde{\sigma}^{(v)}$.
We replace the graph $\Gamma $ by a compact oriented graph $\Gamma _{2}$
that is constructed from the compact part of $\Gamma $ by the following
rule: each internal edge $e_{j}$ of the graph $\Gamma $ is replaced by two
oriented edges (bonds) $b_{j}$, $\hat{b}_{j}$ of the same length as $e_{j}$;
the oriented edges $b_{j}$ and $\hat{b}_{j}$ have the opposite orientations.
On these edges we use the ansatz 
\begin{eqnarray*}
f_{b_{j}}(x) &=&\alpha _{b_{j}}^{\mathrm{in}}\mathrm{e}^{-ikx_{b_{j}}}+%
\alpha _{b_{j}}^{\mathrm{out}}\mathrm{e}^{ikx_{b_{j}}}\,, \\
f_{\hat{b}_{j}}(x) &=&\alpha _{\hat{b}_{j}}^{\mathrm{in}}
\mathrm{e}^{-ikx_{\hat{b}_{j}}}+\alpha _{\hat{b}_{j}}^{\mathrm{out}}%
\mathrm{e}^{ikx_{\hat{b}_{j}}}\,,
\end{eqnarray*}%
where $x_{b_{j}}$ is the coordinate on the bond $b_{j}$ and $x_{\hat{b}_{j}}$
is the coordinate on the bond $\hat{b}_{j}$. Since there is a relation 
$x_{b_{j}}+x_{\hat{b}_{j}}=\ell_{b_{j}}$ with $\ell_{b_{j}} \equiv \ell_j$ being the lengths of
the bond $b_{j}$ or $\hat{b}_{j}$, and since the functional values for both
directed bonds must correspond to each other 
($f_{b_j}(x_{b_j}) = f_{\hat b_j}(\ell_j - x_{\hat b_j})$), we obtain the following
relations between coefficients 
\begin{equation}
\alpha _{b_{j}}^{\mathrm{in}}=\mathrm{e}^{ik\ell _{j}}\alpha_{\hat{b}_{j}}^{\mathrm{out}}\,,
\qquad \alpha _{\hat{b}_{j}}^{\mathrm{in}}=\mathrm{e}^{ik\ell _{j}}\alpha _{b_{j}}^{\mathrm{out}}\,.  
\label{eq-bbhat}
\end{equation}

Now we define several matrices that we will use later.

\begin{definition}
The matrix $\tilde{\Sigma}(k)$ (which is in general energy-dependent) is a
matrix that is similar to a block diagonal matrix with blocks $\tilde{\sigma}_{v}(k)$. 
The similarity transformation is transformation from the basis 
$$
\vec{\alpha}=(\alpha _{b_{1}}^{\mathrm{in}},\dots ,\alpha _{b_{N}}^{\mathrm{in}},
\alpha _{\hat{b}_{1}}^{\mathrm{in}},\dots ,\alpha _{\hat{b}_{N}}^{\mathrm{in}})^{\mathrm{T}} 
$$
to the basis 
$$
(\alpha _{b_{v_{1}1}}^{\mathrm{in}},\dots ,\alpha _{b_{v_{1}d_{1}}}^{\mathrm{in}},
\alpha _{b_{v_{2}1}}^{\mathrm{in}},\dots ,\alpha _{b_{v_{2}d_{2}}}^{\mathrm{in}},
\dots )^{\mathrm{T}}\,, 
$$
where $b_{v_{1}j}$ is the $j$-th edge ending in the vertex $v_{1}$.

Furthermore, we define $2N\times 2N$ matrices $Q = \left(%
\begin{array}{cc}
0 & I_N \\ 
I_N & 0%
\end{array}%
\right)$, scattering matrix $S(k) = Q \tilde\Sigma (k)$ and 
$$
L =\mathrm{diag\,}(\ell_1,\dots , \ell_N,\ell_1,\dots , \ell_N)\,. 
$$
\end{definition}

Using these matrices we can state the following theorem.

\begin{theorem}
\label{thm-rescon} The resonance condition is given by 
$$
\mathrm{det\,}(\mathrm{e}^{ikL} Q \tilde\Sigma(k)- I_{2N}) = 0\,. 
$$
\end{theorem}

\begin{proof}
If we define $\vec{\alpha}_b^{\mathrm{in}} = (\alpha_{b_1}^{\mathrm{in}},\dots \alpha_{b_N}^{\mathrm{in}})^\mathrm{T}$ 
and similarly for the outgoing amplitudes and bonds in the opposite direction 
$\hat b_j$, we can subsequently obtain
$$
  \left(\begin{array}{c}\vec{\alpha}_b^{\mathrm{in}}\\\vec{\alpha}_{\hat{b}}^{\mathrm{in}} \end{array}\right) = \mathrm{e}^{ikL}
\left(\begin{array}{c}\vec{\alpha}_{\hat{b}}^{\mathrm{out}}\\\vec{\alpha}_{b}^{\mathrm{out}} \end{array}\right) = \mathrm{e}^{ikL} Q
\left(\begin{array}{c}\vec{\alpha}_{b}^{\mathrm{out}}\\\vec{\alpha}_{\hat{b}}^{\mathrm{out}} \end{array}\right) =  \mathrm{e}^{ikL} Q
\tilde \Sigma(k)\left(\begin{array}{c}\vec{\alpha}_b^{\mathrm{in}}\\\vec{\alpha}_{\hat{b}}^{\mathrm{in}} \end{array}\right) \,.
$$
First we used relations (\ref{eq-bbhat}) and then definitions of matrices $Q$ 
and $\tilde \Sigma$. Since the vectors on the lhs and rhs are the same, the 
equation of the solvability of the system gives the resonance condition.
\end{proof}

Now we define orbits on the graph using the same notation as in \cite{BHJ}.

\begin{definition}
A \emph{periodic orbit} $\gamma $ on the graph $\Gamma _{2}$ is a closed
path that begins and ends in the same vertex. It can be denoted by the
directed edges that it subsequently contains, e.g. $\gamma
=(b_{1},b_{2},\dots ,b_{n})$. Cyclic permutation of directed bonds does not
change the periodic orbit. A \emph{pseudo orbit} is a collection of periodic
orbits ($\tilde{\gamma}=\{\gamma _{1},\gamma _{2},\dots ,\gamma _{m}\}$). An 
\emph{irreducible pseudo orbit} $\bar{\gamma}$ is a pseudo orbit that
contains no directed bond more than once. The \emph{metric length} of a
periodic orbit is defined as $\ell _{\gamma }=\sum_{b_{j}\in \gamma }\ell
_{b_{j}}$; the length of a pseudo orbit is the sum of the lengths of all
periodic orbits the pseudo orbit is composed of. We denote the product of
scattering amplitudes along the periodic orbit $\gamma =(b_{1},b_{2},\dots
b_{n})$ as $A_{\gamma }=S_{b_{2}b_{1}}S_{b_{3}b_{2}}\dots S_{b_{1}b_{n}}$;
here $S_{b_{i}b_{j}}$ denotes the entry of the matrix $S$ in the row
corresponding to the bond $b_{i}$ and column corresponding to the bond $b_{j}
$. For a pseudo orbit we define $A_{\tilde{\gamma}}=\prod_{\gamma _{j}\in 
\tilde{\gamma}}A_{\gamma _{j}}$. By $m_{\tilde{\gamma}}$ we denote the
number of periodic orbits in the pseudo orbit $\tilde{\gamma}$. The set of
irreducible pseudo orbits also contains irreducible pseudo orbit on zero
edges with $m_{\bar{\gamma}}=0$, $\ell _{\bar{\gamma}}=0$ and $A_{\bar{\gamma}}=1$.
\end{definition}

Without proof, which immediately follows from Theorem~1 in \cite{BHJ}, we
give a theorem on finding the resonance condition using pseudo orbits.

\begin{theorem}
\label{thm-rescon2} The resonance condition is given by 
$$
\sum_{\bar \gamma} (-1)^{m_{\bar \gamma}} A_{\bar \gamma}(k) \,\mathrm{e}^{ik\ell_{\bar \gamma}} = 0\,, 
$$
where the sum goes over all irreducible pseudo orbits $\bar \gamma$.
\end{theorem}

In order to make a summary, let us highlight the main steps of the pseudo
orbit expansion.

\begin{itemize}
\item we construct effective vertex-scattering matrices for each vertex (see
Theorem~\ref{thm-evsm}) and we find the scattering matrix $S(k) = Q \tilde
\Sigma(k)$

\item we construct the oriented graph $\Gamma_2$: we ``cut off'' the
halflines and replace each edge of the compact part of $\Gamma$ by two edges
of opposite directions and the same lengths

\item for each irreducible pseudo orbit we find $m_{\bar\gamma}$, $%
\ell_{\bar\gamma}$ and $A_{\bar\gamma}$ (this is found from the scattering
matrix or directly from the vertex-scattering matrices)

\item we find contribution of each irreducible pseudo orbit on $\Gamma_2$ by
Theorem~\ref{thm-rescon2} and sum them up

\item the contribution of the irreducible pseudo orbit on zero edges is 1
\end{itemize}

For illustration of the procedure we refer to the example in subsection~\ref%
{ex0} or to \cite{Li}.

\section{Effective size of an equilateral graph}

As we stated in section \ref{sec-asym}, the effective size of a graph is
defined as a $\frac{\pi}{2}$-multiple of the coefficient of the leading term
of the asymptotics. In this section we find the effective size of an
equilateral graph from matrices $Q$ and $\tilde \Sigma$.

\begin{definition}
By an equilateral graph we mean a graph that has the same lengths of all
internal edges.
\end{definition}

First, we show a general criterion whether the graph is non-Weyl using the
notions of Theorem \ref{thm-rescon}.

\begin{theorem}
A graph is non-Weyl iff $\mathrm{det\,}\tilde \Sigma (k) = 0$ for all $k\in\hbox{\bb C}$. 
In other words, a graph is non-Weyl iff there exists a vertex
for which $\mathrm{det\,}\tilde\sigma_v(k) = 0$ for all $k\in \hbox{\bb C}$.
\end{theorem}

\begin{proof}
The term of the determinant in Theorem \ref{thm-rescon}, which has the highest 
multiple of $ik$ in the exponent, is 
$\mathrm{det\,}[Q\tilde \Sigma(k)] \,\mathrm{e}^{2ik\sum_{j=1}^N \ell_j}$ and 
the term with the lowest multiple of $ik$ is 1. Theorem 3.1 in \cite{DEL} shows 
that the number of zeros of this determinant in the circle of radius $R$ is 
asymptotically equal to $\frac{2 }{\pi}\mathrm{vol\,}\Gamma$ iff the coefficient 
by the first term is nonzero. Since multiplying by $Q$ means only rearranging 
the rows, $\mathrm{det\,}Q\tilde \Sigma(k) = 0$ iff $\mathrm{det\,}\tilde \Sigma(k) = 0$. 
Since $\tilde \Sigma (k)$ is similar to the matrix with blocks $\tilde\sigma_v$, 
the second part of the theorem follows. 
\end{proof}

From this theorem and Corollary \ref{cor-stan} Theorem 1.2 from \cite{DP}
follows that states that a graph with standard coupling is non-Weyl iff
there is a vertex connecting the same number of finite and infinite edges.

Now we state a theorem that gives the effective size of an equilateral graph.

\begin{theorem}
Let us assume an equilateral graph (all internal edges have lengths $\ell$).
Then the effective size of this graph is $\frac{\ell}{2}n_{\mathrm{nonzero}}$, 
where $n_{\mathrm{nonzero}}$ is the number of nonzero eigenvalues of the
matrix $Q \tilde \Sigma(k)$.
\end{theorem}

\begin{proof}
We use Theorem \ref{thm-rescon} again. First, we notice that the matrix~$Q \tilde \Sigma (k)$ 
can be in this theorem replaced by its Jordan form $D(k) = V Q \tilde \Sigma (k) V^{-1}$ 
with $V$ unitary. The matrix $L$ (defined in section~\ref{sec-pseudo}) is for
an equilateral graph a multiple of the identity matrix, and therefore $V$ 
commutes with $\mathrm{e}^{ik L}$. Unitary transformation does not change the 
determinant of a matrix, hence we have the resonance condition 
$\mathrm{det\,}(\mathrm{e}^{ik\ell} D(k)- I_{2N}) = 0$. Since the matrix under 
the determinant is upper triangular, the determinant is equal to multiplication 
of its diagonal elements. The term of the determinant, which contains the 
highest multiple of $ik$ in the exponent, is $\mathrm{e}^{i k \ell n_{\mathrm{nonzero}}}$. 
The term with the lowest multiple of $ik$ in the exponent is 1. The claim 
follows from Theorem~3.1 in \cite{DEL}. 
\end{proof}

Clearly, if there are $n_{\mathrm{bal}}$ balanced vertices with standard
coupling, then there is at least $n_{\mathrm{bal}}$ zeros in the eigenvalues
of the matrix $Q\tilde{\Sigma}$, and hence the effective size is bounded by 
$W\leq \mathrm{vol\,}\Gamma -\frac{\ell }{2}n_{\mathrm{bal}}$. The following
corollary of the previous theorem gives a criterion saying when this bound
can be improved.

\begin{corollary}
Let $\Gamma$ be an equilateral graph with standard coupling and with $n_\mathrm{bal}$ 
balanced vertices. Then the effective size $W < \mathrm{vol\,}\Gamma - \frac{\ell}{2}n_\mathrm{bal}$ 
iff $\mathrm{rank\,}(Q \tilde \Sigma Q \tilde \Sigma) < \mathrm{rank(Q \tilde \Sigma)} = 2N - n_\mathrm{bal}$.
\end{corollary}

\begin{proof}
For the vertex of degree $d$ the rank of $\tilde \sigma_v$ is either $d-1$ 
for a balanced vertex or $d$ otherwise. Hence 
$\mathrm{rank(Q \tilde \Sigma)} = \mathrm{rank(\tilde \Sigma)} = 2N - n_\mathrm{bal}$. 
The effective size is smaller than $\mathrm{vol\,}\Gamma - \frac{\ell}{2}n_\mathrm{bal}$ 
iff there is at least one 1 just above the diagonal in the block with zeros 
on the diagonal in Jordan form of $Q\tilde \Sigma$. Each Jordan block consisting 
of $n$ zeros on the diagonal and $n-1$ ones above the diagonal has rank $n-1$; 
its square has rank $n-2$. The blocks with eigenvalues other than zero have 
(as well as their squares) rank maximal. Hence the rank of 
$Q \tilde \Sigma Q \tilde \Sigma$ is smaller than the rank of $Q \tilde \Sigma$.
\end{proof}

\section{Deleting edges of the oriented graph}

In this section we develop a method with the help of which the resonance
condition for non-Weyl graphs can be constructed more easily. We restrict
to equilateral graphs with standard coupling. For each balanced vertex we
delete one directed edge of the graph $\Gamma _{2}$, which ends in this
vertex, and replace it by one or several ``ghost
edges''. These edges allow the pseudo orbits to hop from a
vertex to a directed edge that is not connected with the vertex. The
``ghost edges'' do not contribute to the
resonance condition with the coefficient $\mathrm{e}^{ik\ell }$, they only
change the product of scattering amplitudes. The idea referring to this
deleting is a unitary transformation of the matrix $Q\tilde{\Sigma}$, after
which one of its columns consists of all zeros. We use this method in the
following section to prove the main theorems. We explain this method in the
following construction and summarize it in the theorem; for better
understanding the examples in section~\ref{ex1} or in the note \cite{Li}
might be helpful.

\begin{figure}[ht]
\begin{minipage}[b]{0.45\linewidth}
\centering
\includegraphics[width=\textwidth]{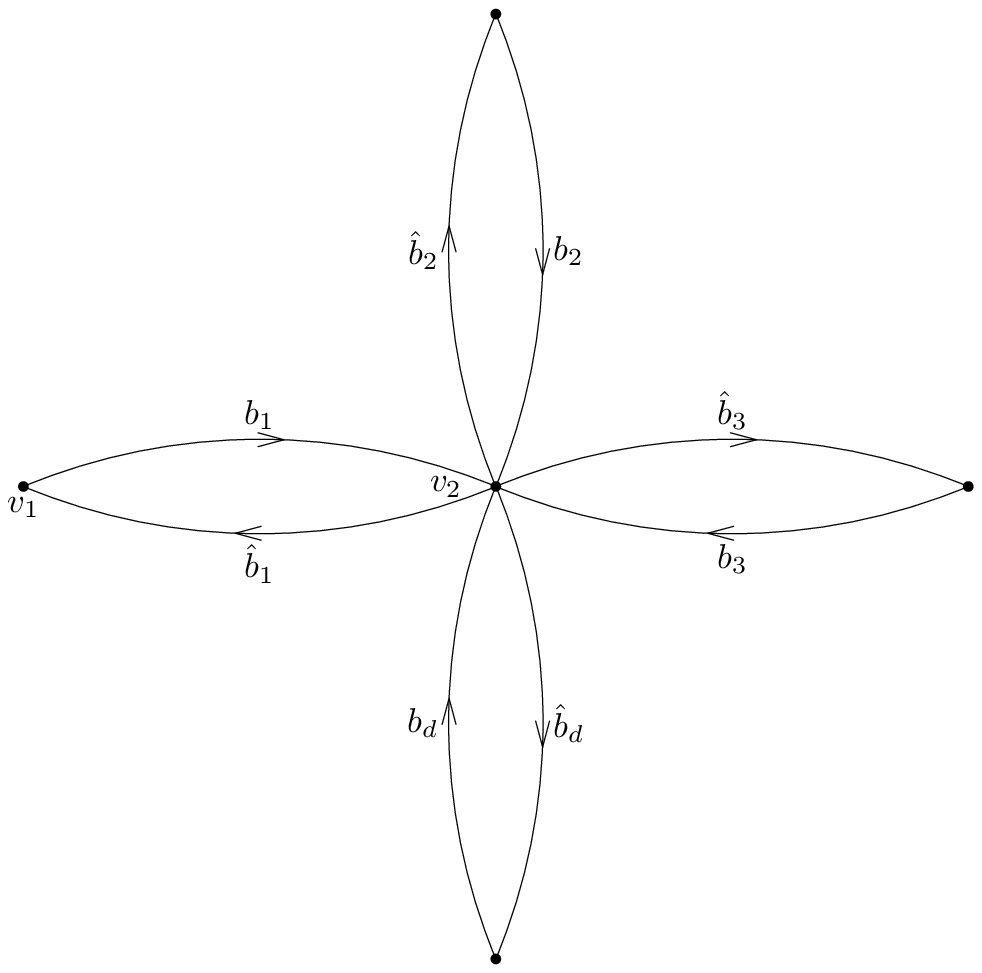}
\caption{Part of the graph $\Gamma_2$. To vertex $v_1$ and to undenoted vertices 
which neighbor $v_2$ other bonds can be attached.}
\label{fig1}
\end{minipage}
\hspace{0.5cm} 
\begin{minipage}[b]{0.45\linewidth}
\centering
\includegraphics[width=\textwidth]{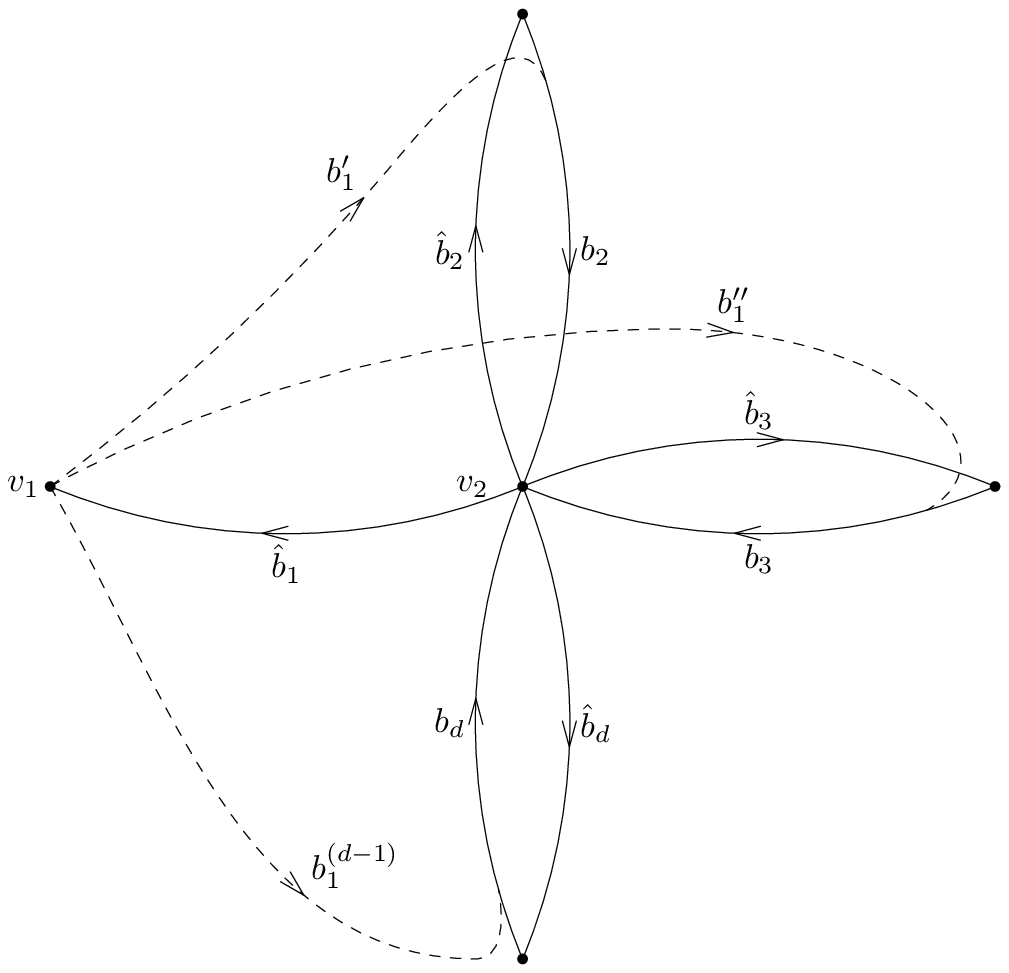}
\caption{Part of the graph $\Gamma_2$ after deleting the bond $b_1$ and 
introducing ``ghost edges''.}
\label{fig2}
\end{minipage}
\end{figure}

Let us assume an equilateral graph $\Gamma $ (all internal edges of 
length~$\ell $) with standard coupling. Let us assume that there is no edge that
starts and ends in one vertex, and that no two vertices are connected by two
or more edges. We use the following construction.

\begin{itemize}
\item let the vertex $v_2$ be balanced in the graph $\Gamma$ and let the
directed bonds $b_1, \dots, b_d$ in the graph $\Gamma_2$ end in the 
vertex~$v_2$ (part of the corresponding directed graph $\Gamma_2$ is shown in 
Figure~\ref{fig1})

\item we delete the directed edge $b_1$ that starts at the vertex $v_1$ and
ends at $v_2$

\item there are new directed ``ghost edges'' $b_1', b_1'', \dots, b^{(d-1)}$ 
introduced that start in the vertex $v_1$ and are connected
to the edges $b_2, b_3, \dots b_d$, respectively (see Figure~\ref{fig2})

\item we use irreducible pseudo orbits to obtain the resonance condition
from the new graph according to Theorem~\ref{thm-rescon2} and obeying the
following rules

\item if the ``ghost edge'' $b_1'$ is contained in the irreducible
pseudo orbit $\bar \gamma$, the length of this ``ghost edge'' does not
contribute to $\ell_{\bar \gamma}$

\item let e.g. the ``ghost edge'' $b_1'$ be included in $\bar
\gamma $; then the scattering amplitude from the bond $b$ ending in $v_1$ to
the bond $b_2$ is the scattering amplitude in the former graph $\Gamma_2$
between $b$ and $b_1$ taken with the opposite sign

\item in the irreducible pseudo orbit each ``ghost edge'' can be used only
once

\item the above procedure can be repeated; for each balanced vertex we can
delete one directed edge that ends in this vertex and replace it by the
``ghost edges''.
\end{itemize}

\begin{theorem}
The described construction under the assumptions above does not change the
resonance condition.
\end{theorem}

\begin{proof}
Unitary transformation of the matrix $Q \tilde \Sigma$ does not change the 
determinant in Theorem~\ref{thm-rescon}, since for an equilateral graph 
$\mathrm{e}^{ik L}$ is replaced by $\mathrm{e}^{ik \ell} I_{2N}$ and 
$$
  \mathrm{det\,}(\mathrm{e}^{ik \ell}V_1^{-1} Q \tilde\Sigma V_1- I_{2N}) = 
  \mathrm{det\,}[V_1^{-1}(\mathrm{e}^{ik \ell} Q \tilde\Sigma - I_{2N})V_1] = 	   
  \mathrm{det\,}(\mathrm{e}^{ik \ell} Q \tilde\Sigma - I_{2N})\,.
$$
We want to delete edge $b_1$. Let $v_2$ be a balanced vertex and let the other 
edges ending in $v_2$ be $b_2, b_3, \dots, b_d$. We choose as $V_1$ the 
$2N\times 2N$ matrix with the following entries
$$
  (V_1)_{i i} =  1\quad \mathrm{for\ }i = 1,\dots, 2N;\quad (V_1)_{b_i b_1} =  
  1\quad \mathrm{for\ }i = 2,\dots, d;\quad (V_1)_{ij} =  0\quad \mathrm{otherwise},   
$$
where $(V_1)_{b_i b_1}$ is the entry with the row corresponding to the bond~$b_i$ 
and column corresponding to $b_1$. One can easily show that the entries of $V_1^{-1}$ are
\begin{eqnarray*}
  (V_1^{-1})_{i i} =  1\quad \mathrm{for\ }i = 1,\dots, 2N;\quad (V_1^{-1})_{b_i b_1} = 
   -1\quad \mathrm{for\ }i = 2,\dots, d;\\
   (V_1^{-1})_{ij} =  0\quad \mathrm{otherwise}.  
\end{eqnarray*}

The matrix $\tilde \sigma_{v_2}= \frac{1}{d}J_{d}-I_d$ has linearly dependent 
columns, hence, if one multiplies it from the right by a matrix with all diagonal 
entries equal to 1, and all the entries in one of its columns equal to 1, other 
entries 0, one obtains a matrix which has one column with all 0s and the other 
columns the same as $\tilde \sigma_{v_2}$. The matrix $Q \tilde \Sigma$ has entries 
of $\tilde \sigma_{v_2}$ in the columns corresponding to bonds ending at $v_2$ 
and rows corresponding to bonds starting from $v_2$. By the same reasoning as 
for $\tilde \sigma_{v_2}$ the column of $Q\tilde \Sigma V_1$ corresponding to 
$b_1$ has all entries equal to 0 and other columns are the same as the 
corresponding columns of $Q\tilde \Sigma$. 

Now it remains to show how multiplying from the left by $V_1^{-1}$ changes the 
matrix. Since nondiagonal entries are only in the rows of $V_1^{-1}$ corresponding 
to bonds $b_2, b_3, \dots b_d$, multiplying by $V_1^{-1}$ changes only these rows. 
Since nondiagonal entries are only in the column corresponding to $b_1$, the 
change can happen only in columns in which there is nontrivial entry in the row 
corresponding to $b_1$. These columns correspond to bonds which end in the 
vertex~$v_1$. We have to multiply these columns by rows of $V_1^{-1}$ which 
have 1 in the $b_j$-th position, $j = 2, 3, \dots, d$ and $-1$ in the $b_1$-th 
position. Since no two vertices are connected by two or more edges, the only 
bond starting at $v_1$ and ending at $v_2$ is $b_1$ and the entries of 
$Q \tilde \Sigma V_1$ in the column corresponding to the edges ending at $v_1$ 
and in the row corresponding to the edges $b_2, b_3, \dots b_d$ are zero (the 
edges in the column cannot be followed in an orbit by edges in the row). Hence,
1 in the above row of $V_1^{-1}$ is multiplied by 0 and $-1$ is multiplied 
by the scattering amplitude between the bonds ending in $v_1$ and $b_1$. 
Therefore, the only change is that there is this scattering amplitude taken 
with the opposite sign in the columns corresponding to the bonds which end 
in $v_1$, and in the row corresponding to bonds $b_2, b_3, \dots, b_d$. These 
entries are represented by the ``ghost edges''. 

It is clear now that one has to take the entry of $I_{2N}$ in the column 
corresponding to $b_1$ in the determinant in Theorem~\ref{thm-rescon} with 
$Q\tilde\Sigma$ replaced by $V_1^{-1} Q \tilde \Sigma V_1$, therefore this edge 
effectively does not exist. The ``ghost edge'' does not contribute to 
$\ell_{\bar \gamma}$, it only says which bonds are connected in the pseudo orbit. 
Similar arguments can be used for other balanced vertices; for each of them we 
delete one edge which ends in it. Note that this method does not delete edges 
to which a ``ghost edge'' leads.
\end{proof}

\section{Main results}

In this section we give two main theorems on bounds on the effective size
for equilateral graphs with standard coupling, and a theorem that gives the
positions of the resonances.

\begin{theorem}
\label{thm-main1} Let us assume an equilateral graph with $N$ internal edges
of lengths $\ell $, with standard coupling, $n_{\mathrm{bal}}$ balanced
vertices, and $n_{\mathrm{nonneig}}$ balanced vertices that do not neighbor
any other balanced vertex. Then the effective size is bounded by $W\leq
N\ell -\frac{\ell }{2}n_{\mathrm{bal}}-\frac{\ell }{2}n_{\mathrm{nonneig}}$.
\end{theorem}

\begin{proof}
Clearly, for each balanced vertex, we can delete one directed edge of the 
graph~$\Gamma_2$, the size of the graph is reduced by $\frac{\ell}{2}n_{\mathrm{bal}}$. 
In the balanced vertex of the degree $d$ that does not neighbor any other 
balanced vertex we have $d-1$ incoming directed bonds and $d$ outgoing directed 
bonds. No outgoing bond is deleted and no ``ghost edge'' ends in the outgoing 
edge, because there is no balanced vertex which neighbors the given vertex. 
Hence, we cannot use one of the outgoing directed edges in the irreducible 
pseudo orbit (there is no way to get this vertex for the $d$-th time). The 
longest irreducible pseudo orbit does not include $n_{\mathrm{nonneig}}$ bonds 
and the effective size of the graph must be reduced by $\frac{\ell}{2}n_{\mathrm{nonneig}}$.
\end{proof}

\begin{theorem}
\label{thm-square} Let us assume an equilateral graph ($N$ internal edges of
the lengths $\ell $) with standard coupling. Let there be a square of
balanced vertices $v_{1}$, $v_{2}$, $v_{3}$ and $v_{4}$ without diagonals,
i.e. $v_{1}$ neighbors $v_{2}$, $v_{2}$ neighbors $v_{3}$, $v_{3}$ neighbors 
$v_{4}$, $v_{4}$ neighbors $v_{1}$, $v_{1}$ does not neighbor $v_{3}$ and 
$v_{2}$ does not neighbor $v_{4}$. Then the effective size is bounded by 
$W\leq (N-3)\ell $.
\end{theorem}

\begin{figure}[tbp]
\begin{minipage}[b]{0.45\linewidth}
\centering
\includegraphics[width=\textwidth]{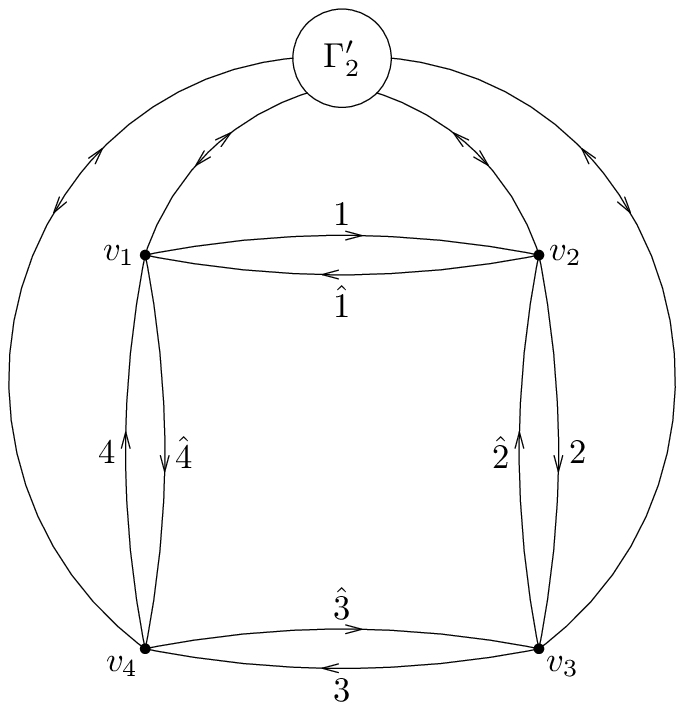}
\caption{Figure to Theorem \ref{thm-square}. The bonds between $\Gamma_2'$ and 
vertices $v_1, \dots, v_4$ represent possible directed edges between this 
subgraph and vertices in both directions.}
\label{fig3}
\end{minipage}
\hspace{0.5cm} 
\begin{minipage}[b]{0.45\linewidth}
\centering
\includegraphics[width=\textwidth]{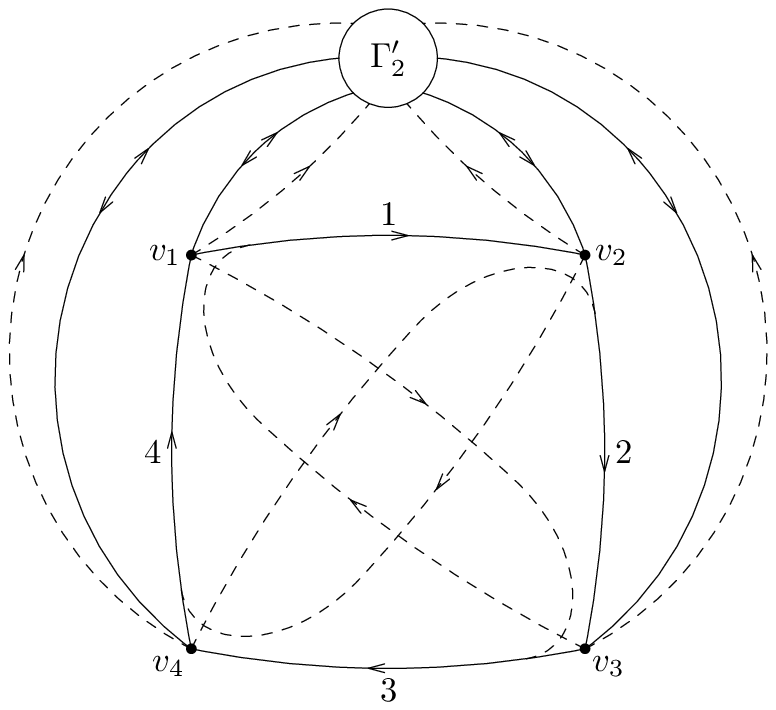}
\caption{Figure to Theorem \ref{thm-square}. The bonds and ``ghost edges'' 
between $\Gamma_2'$ and vertices $v_1, \dots, v_4$ represent possible directed 
edges between this subgraph and vertices in both directions and ``ghost edges'' 
from the vertices $v_1,\dots, v_4$ to the subgraph $\Gamma_2'$.}
\label{fig4}
\end{minipage}
\end{figure}

\begin{figure}[tbp]
\centering
\includegraphics[width=0.45\textwidth]{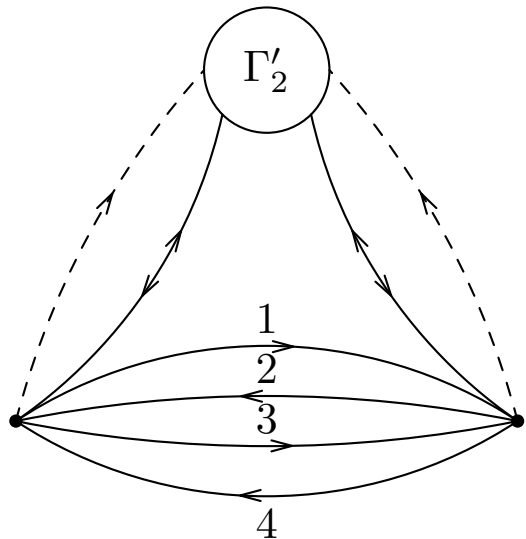}
\caption{Figure to Theorem \protect\ref{thm-square}. The effective graph
with ``ghost edges'' taken into account.}
\label{fig5}
\end{figure}

\begin{proof}
Let us denote the bond from $v_1$ to $v_2$ by 1, the bond from $v_2$ to $v_3$ 
by 2, the bond from $v_3$ to $v_4$ by 3 and the bond from $v_4$ to $v_1$ by 4, 
the bonds in the opposite directions by $\hat 1, \hat 2, \hat 3$ and $\hat 4$ 
(see Figure~\ref{fig3}). Let $\Gamma_2'$ be the rest of the graph $\Gamma_2$; 
it can be connected with the square by bonds in both directions, we denote them 
in Figure~\ref{fig3} by edges with arrows in both directions. Now we delete 
bonds $\hat 1, \hat 2, \hat 3$ and $\hat 4$ (see Figure~\ref{fig4}), there 
arise ``ghost edges'' in the square (explicitly shown), and there may arise 
``ghost edges'' from vertices of the square to edges of the rest of the graph 
(represented by dashed edges between the square and $\Gamma_2'$). Since the 
pseudo orbit can continue from the bond 1 to the bond 4 (with scattering 
amplitude equal to the scattering amplitude from 1 to $\hat 1$ with the opposite 
sign), from the bond 2 to the bond 1, etc., one can effectively represent the 
oriented graph by Figure~\ref{fig5}. 

It is clear that the effective size has been reduced by $2\ell$, because four 
edges have been deleted. The term of the resonance condition with the highest 
multiple of $ik$ in the exponent corresponds to the contribution of irreducible 
pseudo orbits on all remaining ``non-ghost'' bonds (the pseudo orbits may or 
may not use the ``ghost edges''). The contribution of the pseudo orbits $(1234)$ 
and $(12)(34)$ cancels out, because both pseudo orbits differ only in the number 
of orbits, hence there is a factor of $-1$. A similar argument holds also for the 
pair of pseudo orbits $(1432)$ and $(14)(32)$ and all irreducible pseudo orbits 
that include these pseudo orbits. 

Now we show why also the second highest term is zero. It includes the 
contribution of all bonds but one. If the non-included bond is not 1, 2, 3 or 4, 
the contributions are canceled due to the previous argument. If e.g. the bond 4 
is not included, it would mean that one must go from one of the lower vertices 
in Figure~\ref{fig5} to the other through $\Gamma_2'$. This is not possible 
because the irreducible pseudo orbit has to include all the bonds but 4, none 
of the bonds in part $\Gamma_2'$ is deleted and there is no ``ghost edge'' 
ending in the bond 1, 2, 3 or 4. If there exists a path through $\Gamma_2'$, 
from one vertex to another, then the path in the opposite direction cannot be 
covered by the irreducible pseudo orbit. 

Therefore, at least 6 directed edges of the former graph $\Gamma_2$ are not used 
and the effective size is reduced by $3 \ell$.
\end{proof}

Finally, we state what the positions of the resolvent resonances are.

\begin{theorem}
\label{thm-expliciteres} Let us assume an equilateral graph (lengths $\ell$)
with standard coupling. Let the eigenvalues of $Q \tilde \Sigma$ be $c_j =
r_j \mathrm{e}^{i\varphi_j}$. Then the resolvent resonances are $\lambda=
k^2 $ with $k = \frac{1}{\ell}(-\varphi_j+ 2n\pi+i\ln{r_j})$, $n\in\hbox{\bb Z}$. 
Moreover, $|c_j|\leq 1$ and for a graph with no edge starting
and ending in one vertex also $\sum_{j=1}^{2N} c_j = 0$.
\end{theorem}

\begin{proof}
The resonance condition is
$$
  \prod_{j=1}^{2N} (\mathrm{e}^{ik\ell} c_j -1) = 0\,,
$$ 
hence we have for $k = k_{\mathrm{R}}+ ik_{\mathrm{I}}$
$$
  r_j \mathrm{e}^{-k_{\mathrm{I}}\ell}\mathrm{e}^{ik_{\mathrm{R}} \ell}\mathrm{e}^{i\varphi_j} = 1\,,
$$
from which the claim follows. For $r_j>1$ we would have positive imaginary part 
of $k$ which would contradict the fact that eigenvalues of the selfadjoint 
Hamiltonian are real (the corresponding generalized eigenfunction would be 
square integrable). If the graph does not have any edge starting and ending 
in one vertex, then there are zeros on the diagonal of $Q \tilde \Sigma$, hence 
its trace (the sum of its eigenvalues) is zero. 
\end{proof}

\section{Examples}

In this section we show three particular examples that illustrate the
general behavior. In the first example the methods of pseudo orbit expansion
and deleting edges are explained. In the second example we deal with
deleting the directed bonds again. The third example shows that the symmetry
of the graph is not sufficient in order to obtain the effective size smaller
than it is expected from the bound in Theorem~\ref{thm-main1}.

\subsection{Three abscissas and three halflines}

\label{ex0} In this example we show in detail how the method of pseudo orbit
expansion and the method of ``deleting
edges'' are used. Let us consider a tree graph with three
internal edges of lengths $\ell $ and three halflines attached in the
central vertex (see Figure~\ref{fig11}). There is standard coupling in the
central vertex and Dirichlet coupling at the spare ends of the internal
edges.

\begin{figure}[ht]
\centering
\includegraphics[width=0.35\textwidth]{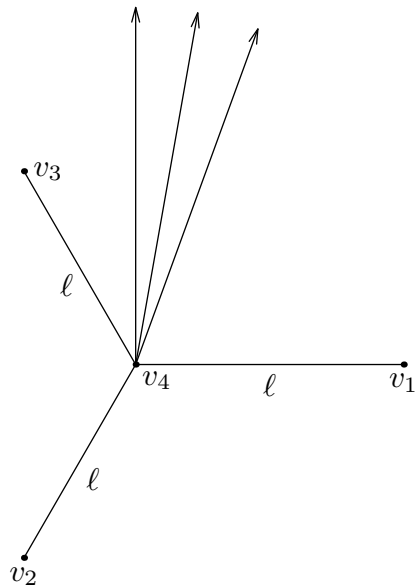}
\caption{Graph $\Gamma$ in example~\protect\ref{ex0}.}
\label{fig11}
\end{figure}

\begin{figure}[ht]
\begin{minipage}[b]{0.45\linewidth}
\centering
\includegraphics[width=\textwidth]{fig12}
\caption{The graph $\Gamma_2$ in example~\ref{ex0}.}
\label{fig12}
\end{minipage}
\hspace{0.5cm} 
\begin{minipage}[b]{0.45\linewidth}
\centering
\includegraphics[width=\textwidth]{fig13}
\caption{The graph $\Gamma_2$ after deleting the bond $3$.}
\label{fig13}
\end{minipage}
\end{figure}

Since there is Dirichlet coupling for the vertices $v_1$, $v_2$ and $v_3$,
one has $U_v = -1$ for $v = v_1,\ v_2,\ v_3$. From (\ref{eq-utilde}) we have 
$\tilde U_v = -1$ and by Theorem~\ref{thm-evsm} $\tilde \sigma_v = -1$.
Since there is standard coupling in the middle vertex, we have from
Corollary~\ref{cor-stan} $\sigma_{v_4} = \frac{1}{3} J_3 - I_3 = \frac{1}{3}%
\left(%
\begin{array}{ccc}
-2 & 1 & 1 \\ 
1 & -2 & 1 \\ 
1 & 1 & -2%
\end{array}%
\right)$.

The graph $\Gamma _{2}$ is obtained by ``cutting
off'' the halflines and replacing each edge of the compact
part by two directed bonds (see Figure~\ref{fig12}). One can easily see that
there are no irreducible pseudo orbits on the odd number of edges. There is
one (trivial) pseudo orbit on zero edges. There are three irreducible pseudo
orbits on two edges $(1\hat{1})$; $(2\hat{2})$ and $(3\hat{3})$. For each of
these pseudo orbits we have $A_{\bar{\gamma}}=(-1)\left( -\frac{2}{3}\right) 
$ (there is $-1$ for the scattering from $\hat{1}$ to $1$ and diagonal entry
of the effective vertex-scattering matrix $-2/3$ for scattering from $1$ to 
$\hat{1}$), $m_{\bar{\gamma}}=1$ and $\ell _{\bar{\gamma}}=2\ell $. We have
six irreducible pseudo orbits on four edges $(1\hat{1})(2\hat{2})$; 
$(1\hat{1})(3\hat{3})$; $(2\hat{2})(3\hat{3})$; $(1\hat{2}2\hat{1})$; 
$(1\hat{3}3\hat{1})$ and $(2\hat{3}3\hat{2})$. And there are six irreducible 
pseudo orbits on six edges $(1\hat{1})(2\hat{2})(3\hat{3})$; 
$(1\hat{2}2\hat{1})(3\hat{3})$; $(1\hat{3}3\hat{1})(2\hat{2})$; 
$(2\hat{3}3\hat{2})(1\hat{1})$; $(1\hat{2}2\hat{3}3\hat{1})$ and 
$(1\hat{3}3\hat{2}2\hat{1})$. The first one consists
of three periodic orbits, the next three of two periodic orbits, and the
last two of one periodic orbit. The coefficients by the exponentials are
hence the following 
\begin{eqnarray*}
\mathrm{exp\,}0 &:&1\,, \\
\mathrm{exp\,}(2ik\ell ) &:&(-1)\left( -\frac{2}{3}\right) (-1)^{1}\cdot
3=-2\,, \\
\mathrm{exp\,}(4ik\ell ) &:&(-1)^{2}\left( -\frac{2}{3}\right)
^{2}(-1)^{2}\cdot 3+(-1)^{2}\left( \frac{1}{3}\right) ^{2}(-1)^{1}\cdot
3=1\,, \\
\mathrm{exp\,}(6ik\ell ) &:&(-1)^{3}\left( -\frac{2}{3}\right)
^{3}(-1)^{3}+(-1)^{3}\left( \frac{1}{3}\right) ^{2}\left( -\frac{2}{3}\right) (-1)^{2}\cdot 3+ \\
&&+(-1)^{3}\left( \frac{1}{3}\right) ^{3}(-1)^{1}\cdot 2=0\,. \\
\end{eqnarray*}

Since the vertex $v_{4}$ was balanced, we can delete one directed bond, e.g.
the bond~$3$. The graph $\Gamma _{2}$ after deleting the edge and
introducing ``ghost edges'' is shown in
Figure~\ref{fig13}. One can easily see that again we do not have irreducible
pseudo orbits on an even number of ``non-ghost'' edges. We again have the trivial irreducible
pseudo orbit on zero edges. We have the following four irreducible pseudo
orbits on two ``non-ghost'' edges: $(1\hat{1})$; 
$(2\hat{2})$; $(\hat{3}3'1)$ and $(\hat{3}3'' 2)$. 
The scattering amplitude for the orbit $(\hat{3}3'1)$ from $\hat{3}$ 
via $3'$ to $1$ is $+1$, because the scattering amplitude in the
former graph from $\hat{3}$ to $3$ was $-1$. Finally, we have the following
five irreducible pseudo orbits on four ``non-ghost'' edges: $(1\hat{1})(2\hat{2})$; 
$(1\hat{1})(\hat{3}3'' 2)$; $(2\hat{2})(\hat{3}3'1)$; 
$(1\hat{2}2\hat{1})$; $(1\hat{3}3'' 2\hat{1})$ and 
$(2\hat{3}3'1\hat{2})$. Clearly, there are no pseudo orbits on six 
edges, since we have deleted one edge. The coefficient by the exponentials 
in the resonance condition are 
\begin{eqnarray*}
\mathrm{exp\,}0 &:&1\,, \\
\mathrm{exp\,}(2ik\ell ) &:&(-1)\left( -\frac{2}{3}\right) (-1)^{1}\cdot 2+1\frac{1}{3}(-1)^{1}\cdot 2=-2\,, \\
\mathrm{exp\,}(4ik\ell ) &:&(-1)^{2}\left( -\frac{2}{3}\right)
^{2}(-1)^{2}+(-1)\left( -\frac{2}{3}\right) 1\frac{1}{3}(-1)^{2}\cdot 2+ \\
&&+(-1)^{2}\left( \frac{1}{3}\right) ^{2}(-1)^{1}+(-1)\left( \frac{1}{3}\right) ^{2}1(-1)^{1}\cdot 2=1\,,
\end{eqnarray*}

In both cases the resonance condition is $1-2\,\mathrm{e}^{2ik\ell}+\mathrm{e}^{4ik\ell} = 0$. 
There are resonances at $k^2$, $k = \frac{n\pi}{2\ell}$, 
$n\in \hbox{\bb Z}$ with multiplicity 2.

\subsection{Square with the diagonal and all vertices balanced} \label{ex1} 
We assume an equilateral graph with four internal edges
in the square and the fifth edge as diagonal of this square. All four
vertices are balanced, i.e. to two of them two halflines are attached and to
the other two vertices three halflines are attached. There is standard
coupling in all vertices. The oriented graph~$\Gamma_2$ is shown in 
Figure~\ref{fig6}. 

\begin{figure}[ht]
\begin{minipage}[b]{0.45\linewidth}
\centering
\includegraphics[width=\textwidth]{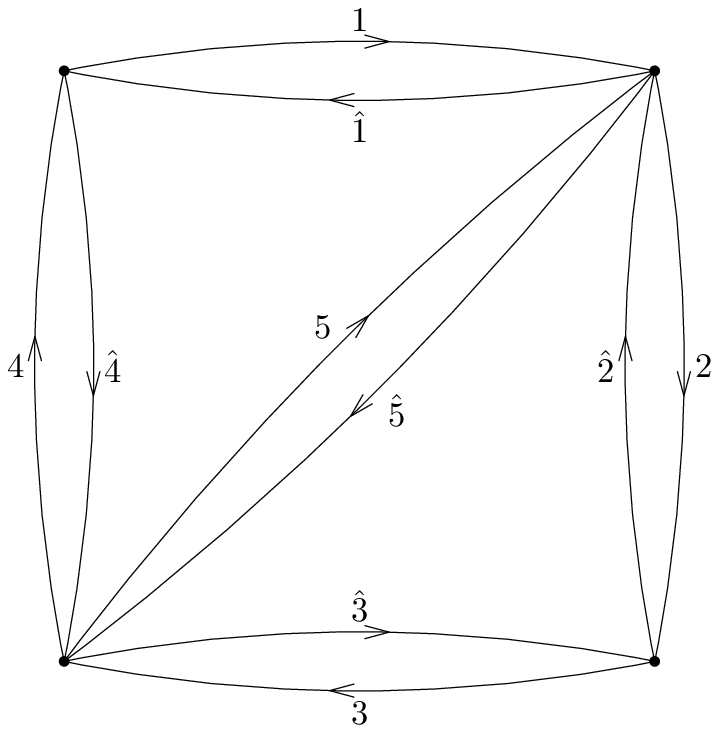}
\caption{The graph $\Gamma_2$ for the graph in example~\ref{ex1}.}
\label{fig6}
\end{minipage}
\hspace{0.5cm}
\begin{minipage}[b]{0.45\linewidth}
\centering
\includegraphics[width=\textwidth]{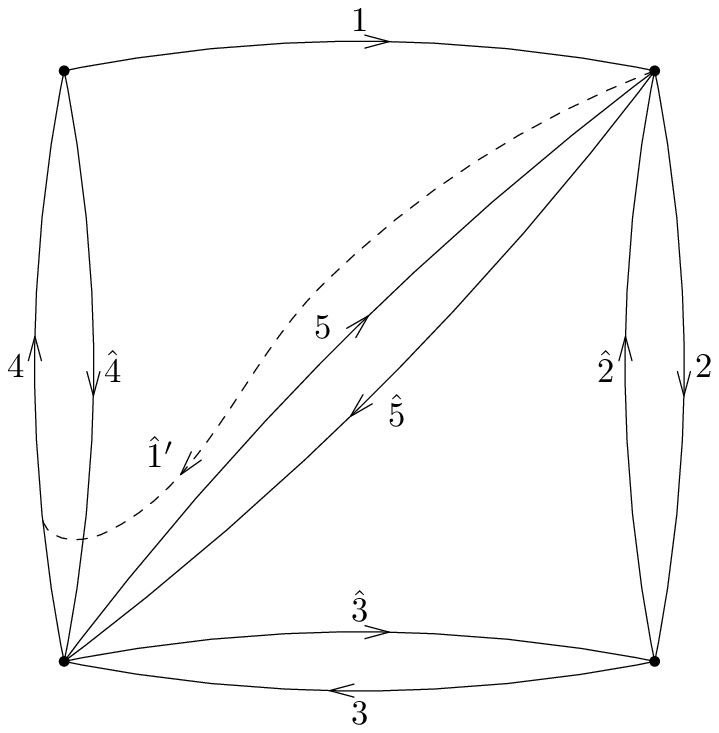}
\caption{The graph $\Gamma_2$ after deleting the bond $\hat 1$.}
\label{fig7}
\end{minipage}
\end{figure} 

Let us denote the vertex from which the edge $1$ starts by $v_1$ and
the vertex, where $1$ ends by $v_2$. Then the effective vertex-scattering
matrices are 
$$
\tilde\sigma_{v_1} = \frac{1}{2}\left(%
\begin{array}{cc}
-1 & 1 \\ 
1 & -1%
\end{array}%
\right)\,,\quad \tilde\sigma_{v_2} = \frac{1}{3}\left(%
\begin{array}{ccc}
-2 & 1 & 1 \\ 
1 & -2 & 1 \\ 
1 & 1 & -2%
\end{array}%
\right)\,, 
$$
similarly for the other two vertices. Hence, the matrix $Q\tilde \Sigma$ is
equal to 
$$
\begin{array}{c|cccccccccc}
& 1 & 2 & 3 & 4 & 5 & \hat 1 & \hat 2 & \hat 3 & \hat 4 & \hat 5 \\ \hline
1 & 0 & 0 & 0 & 1/2 & 0 & -1/2 & 0 & 0 & 0 & 0 \\ 
2 & 1/3 & 0 & 0 & 0 & 1/3 & 0 & -2/3 & 0 & 0 & 0 \\ 
3 & 0 & 1/2 & 0 & 0 & 0 & 0 & 0 & -1/2 & 0 & 0 \\ 
4 & 0 & 0 & 1/3 & 0 & 0 & 0 & 0 & 0 & -2/3 & 1/3 \\ 
5 & 0 & 0 & 1/3 & 0 & 0 & 0 & 0 & 0 & 1/3 & -2/3 \\ 
\hat 1 & -2/3 & 0 & 0 & 0 & 1/3 & 0 & 1/3 & 0 & 0 & 0 \\ 
\hat 2 & 0 & -1/2 & 0 & 0 & 0 & 0 & 0 & 1/2 & 0 & 0 \\ 
\hat 3 & 0 & 0 & -2/3 & 0 & 0 & 0 & 0 & 0 & 1/3 & 1/3 \\ 
\hat 4 & 0 & 0 & 0 & -1/2 & 0 & 1/2 & 0 & 0 & 0 & 0 \\ 
\hat 5 & 1/3 & 0 & 0 & 0 & -2/3 & 0 & 1/3 & 0 & 0 & 0 
\end{array}%
\,. 
$$
The edges, to which the rows and columns correspond, are denoted on the left
and on the top of the matrix.

Now we delete the bond $\hat 1$ (Figure~\ref{fig7}). Deleting this
edge is equivalent to the unitary transformation $V_1^{-1} Q \tilde \Sigma
V_1$, where $V_1$ has 1 on the diagonal, 1 in the sixth column
(corresponding to the bond $\hat 1$) and the fourth row (corresponding to
the bond 4), other entries of this matrix are zero ($(V_1)_{ii} = 1$, 
$\forall i$, $(V_1)_{4\hat 1} = 1$, $(V_1)_{ij} =0$ otherwise). Its inverse
has 1 on the diagonal and $-1$ in the sixth column and fourth row. We obtain
the matrix 
$$
\begin{array}{c|cccccccccc}
& 1 & 2 & 3 & 4 & 5 & \hat 1 & \hat 2 & \hat 3 & \hat 4 & \hat 5 \\ \hline
1 & 0 & 0 & 0 & 1/2 & 0 & 0 & 0 & 0 & 0 & 0 \\ 
2 & 1/3 & 0 & 0 & 0 & 1/3 & 0 & -2/3 & 0 & 0 & 0 \\ 
3 & 0 & 1/2 & 0 & 0 & 0 & 0 & 0 & -1/2 & 0 & 0 \\ 
4 & \mathbf{2/3} & 0 & 1/3 & 0 & \mathbf{-1/3} & 0 & \mathbf{-1/3} & 0 & -2/3
& 1/3 \\ 
5 & 0 & 0 & 1/3 & 0 & 0 & 0 & 0 & 0 & 1/3 & -2/3 \\ 
\hat 1 & -2/3 & 0 & 0 & 0 & 1/3 & 0 & 1/3 & 0 & 0 & 0 \\ 
\hat 2 & 0 & -1/2 & 0 & 0 & 0 & 0 & 0 & 1/2 & 0 & 0 \\ 
\hat 3 & 0 & 0 & -2/3 & 0 & 0 & 0 & 0 & 0 & 1/3 & 1/3 \\ 
\hat 4 & 0 & 0 & 0 & -1/2 & 0 & 0 & 0 & 0 & 0 & 0 \\ 
\hat 5 & 1/3 & 0 & 0 & 0 & -2/3 & 0 & 1/3 & 0 & 0 & 0 
\end{array}%
\,. 
$$
In this matrix the sixth column consists of all zeros and there are three
other new entries in the fourth row that are printed in bold. Since the new
entries are in the fourth row, there must be a ``ghost edge'' $\hat
1'$ from the vertex $v_2$ pointing to the bond 4. We can use it in
pseudo orbits containing one of the bonds 1, 5, $\hat 2$ continuing then to
the bond 4 with the scattering amplitudes in bold above.

\begin{figure}[ht]
\begin{minipage}[b]{0.45\linewidth}
\centering
\includegraphics[width=\textwidth]{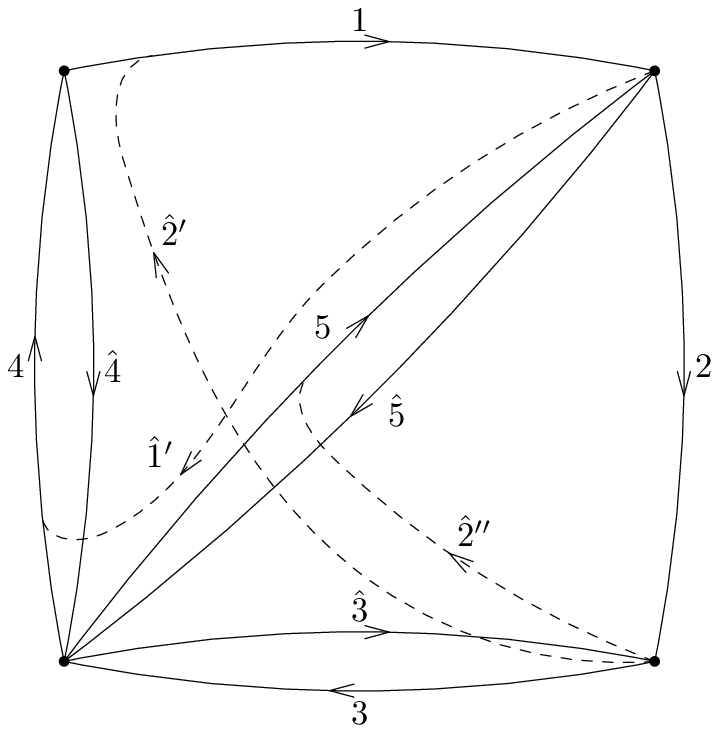}
\caption{The graph $\Gamma_2$ after deleting the bonds $\hat 1$ and $\hat 2$.}
\label{fig8}
\end{minipage}
\hspace{0.5cm}
\begin{minipage}[b]{0.45\linewidth}
\centering
\includegraphics[width=\textwidth]{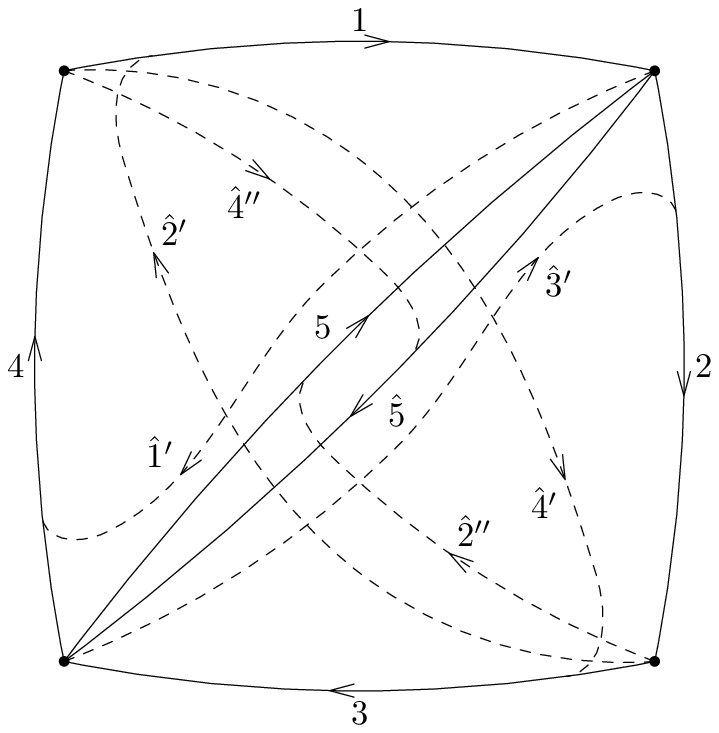}
\caption{The graph $\Gamma_2$ after deleting the bonds $\hat 1$, $\hat 2$, $\hat 3$ and~$\hat 4$.}
\label{fig9}
\end{minipage}
\end{figure} 

Now we delete the bond $\hat{2}$ (Figure~\ref{fig8}). Since two
other bonds end in the vertex $v_{2}$ (bonds 1 and 5), there will be two
``ghost edges''. We use the unitary
transformation $V_{2}^{-1}V_{1}^{-1}Q\tilde{\Sigma}V_{1}V_{2}$. $V_{2}$ has
1 on the diagonal and 1 in the seventh column (corresponding to the edge 
$\hat{2}$) and rows 1 and 5. Its inverse has nondiagonal terms with opposite
signs. After the transformation we obtain 
$$
\begin{array}{c|cccccccccc}
& 1 & 2 & 3 & 4 & 5 & \hat{1} & \hat{2} & \hat{3} & \hat{4} & \hat{5} \\ 
\hline
1 & 0 & \mathbf{1/2} & 0 & 1/2 & 0 & 0 & 0 & \mathbf{-1/2} & 0 & 0 \\ 
2 & 1/3 & 0 & 0 & 0 & 1/3 & 0 & 0 & 0 & 0 & 0 \\ 
3 & 0 & 1/2 & 0 & 0 & 0 & 0 & 0 & -1/2 & 0 & 0 \\ 
4 & 2/3 & 0 & 1/3 & 0 & -1/3 & 0 & 0 & 0 & -2/3 & 1/3 \\ 
5 & 0 & \mathbf{1/2} & 1/3 & 0 & 0 & 0 & 0 & {-\mathbf{1/2}} & 1/3 & -2/3 \\ 
\hat{1} & -2/3 & 0 & 0 & 0 & 1/3 & 0 & 0 & 0 & 0 & 0 \\ 
\hat{2} & 0 & -1/2 & 0 & 0 & 0 & 0 & 0 & 1/2 & 0 & 0 \\ 
\hat{3} & 0 & 0 & -2/3 & 0 & 0 & 0 & 0 & 0 & 1/3 & 1/3 \\ 
\hat{4} & 0 & 0 & 0 & -1/2 & 0 & 0 & 0 & 0 & 0 & 0 \\ 
\hat{5} & 1/3 & 0 & 0 & 0 & -2/3 & 0 & 0 & 0 & 0 & 0 
\end{array}%
\,. 
$$
The seventh column (corresponding to $\hat{2}$) has all entries equal to
zero; other new entries are printed in bold. These entries correspond to the
new ``ghost edges'' $\hat{2}'$ and 
$\hat{2}''$. Similarly, we delete edges $\hat{3}$ and $\hat{4}$ 
(Figure~\ref{fig9}); the matrix $Q\tilde{\Sigma}$ after these
transformations is 
$$
\begin{array}{c|cccccccccc}
& 1 & 2 & 3 & 4 & 5 & \hat{1} & \hat{2} & \hat{3} & \hat{4} & \hat{5} \\ 
\hline
1 & 0 & 1/2 & 0 & 1/2 & 0 & 0 & 0 & 0 & 0 & 0 \\ 
2 & 1/3 & 0 & 2/3 & 0 & 1/3 & 0 & 0 & 0 & 0 & -1/3 \\ 
3 & 0 & 1/2 & 0 & 1/2 & 0 & 0 & 0 & 0 & 0 & 0 \\ 
4 & 2/3 & 0 & 1/3 & 0 & -1/3 & 0 & 0 & 0 & 0 & 1/3 \\ 
5 & 0 & 1/2 & 1/3 & 0 & 0 & 0 & 0 & 0 & 0 & -2/3 \\ 
\hat{1} & -2/3 & 0 & 0 & 0 & 1/3 & 0 & 0 & 0 & 0 & 0 \\ 
\hat{2} & 0 & -1/2 & 0 & 0 & 0 & 0 & 0 & 0 & 0 & 0 \\ 
\hat{3} & 0 & 0 & -2/3 & 0 & 0 & 0 & 0 & 0 & 0 & 1/3 \\ 
\hat{4} & 0 & 0 & 0 & -1/2 & 0 & 0 & 0 & 0 & 0 & 0 \\ 
\hat{5} & 1/3 & 0 & 0 & 1/2 & -2/3 & 0 & 0 & 0 & 0 & 0 
\end{array}%
\,. 
$$
The eigenvalues of the matrix $Q\tilde{\Sigma}$ are $-2/3$, $-1/3$, $-1$, 1
with multiplicity 2, 0 with multiplicity 5. There is one Jordan block 
$\left( 
\begin{array}{cc}
0 & 1 \\ 
0 & 0%
\end{array}%
\right) $ in the Jordan form of $Q\tilde{\Sigma}$. The resonance condition
can be obtained from the eigenvalues by the equation at the beginning of the
proof of Theorem~\ref{thm-expliciteres} 
$$
1-\frac{16}{9}\mathrm{e}^{2ik\ell }-\frac{2}{9}\mathrm{e}^{3ik\ell }+
\frac{7}{9}\mathrm{e}^{4ik\ell }+\frac{2}{9}\mathrm{e}^{5ik\ell }=0\,. 
$$
By Theorem~\ref{thm-expliciteres} we obtain that the positions of the
resonances are such $\lambda =k^{2}$ with $k=\frac{1}{\ell }[(2n+1)\pi -i\ln 
{3}]$, $k=\frac{1}{\ell }[(2n+1)\pi -i\ln {\frac{3}{2}}]$, 
$k=\frac{1}{\ell }(2n+1)\pi $ and $k=\frac{1}{\ell }2n\pi $ with multiplicity 2, 
$n\in \hbox{\bb Z}$.

\subsection{Fully connected graph on four vertices with all vertices
balanced}

\begin{figure}[ht]
\centering
\includegraphics[width=0.45\textwidth]{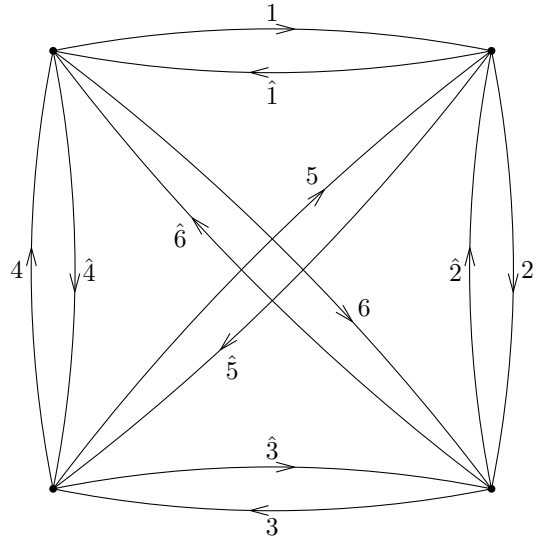}
\caption{Graph $\Gamma_2$ for fully connected graph on four vertices.}
\label{fig10}
\end{figure} 

In this subsection we assume a graph on four vertices, every two
vertices are connected by one edge of length $\ell$ and there are three
halflines attached at each vertex, hence each vertex is balanced. The
directed graph $\Gamma_2$ is shown in Figure~\ref{fig10}. For each vertex we
have the effective vertex-scattering matrix 
$$
\tilde\sigma_v = \frac{1}{3}\left(%
\begin{array}{ccc}
-2 & 1 & 1 \\ 
1 & -2 & 1 \\ 
1 & 1 & -2%
\end{array}%
\right)\,. 
$$
The matrix $Q \tilde \Sigma$ is 
$$
  \begin{array}{c|cccccccccccc}&1&2&3&4&5&6&\hat 1&\hat 2&\hat 3&\hat 4&\hat 5&\hat 6\\ \hline
       1 & 0 & 0 & 0 & 1/3 & 0 & 0 & -2/3 & 0 & 0 & 0 & 0 & 1/3\\
       2 & 1/3 & 0 & 0 & 0 & 1/3 & 0 & 0 & -2/3 & 0 & 0 & 0 & 0\\
       3 & 0 & 1/3 & 0 & 0 & 0 & 1/3 & 0 & 0 & -2/3 & 0 & 0 & 0\\
       4 & 0 & 0 & 1/3 & 0 & 0 & 0 & 0 & 0 & 0 & -2/3 & 1/3 & 0\\
       5 & 0 & 0 & 1/3 & 0 & 0 & 0 & 0 & 0 & 0 & 1/3 & -2/3 & 0\\
       6 & 0 & 0 & 0 & 1/3 & 0 & 0 & 1/3 & 0 & 0 & 0 & 0 & -2/3\\
  \hat 1 & -2/3 & 0 & 0 & 0 & 1/3 & 0 & 0 & 1/3 & 0 & 0 & 0 & 0\\
  \hat 2 & 0 & -2/3 & 0 & 0 & 0 & 1/3 & 0 & 0 & 1/3 & 0 & 0 & 0\\
  \hat 3 & 0 & 0 & -2/3 & 0 & 0 & 0 & 0 & 0 & 0 & 1/3 & 1/3 & 0\\
  \hat 4 & 0 & 0 & 0 & -2/3 & 0 & 0 & 1/3 & 0 & 0 & 0 & 0 & 1/3\\
  \hat 5 & 1/3 & 0 & 0 & 0 & -2/3 & 0 & 0 & 1/3 & 0 & 0 & 0 & 0\\
  \hat 6 & 0 & 1/3 & 0 & 0 & 0 & -2/3 & 0 & 0 & 1/3 & 0 & 0 & 0\\
  \end{array}\,.
$$
Its eigenvalues are $-1$ with multiplicity 2, 1 with multiplicity 3, $-1/3$
with multiplicity 3 and 0 with multiplicity 4. There is no Jordan block. The
resonance condition is 
$$
1- \frac{8}{3}\mathrm{e}^{2ik\ell}-\frac{8}{27}\mathrm{e}^{3ik\ell}+
\frac{62}{27}\mathrm{e}^{4ik\ell}+\frac{16}{27}\mathrm{e}^{5ik\ell}
-\frac{16}{27}\mathrm{e}^{6ik\ell}-\frac{8}{27}\mathrm{e}^{7ik\ell}
- \frac{1}{27}\mathrm{e}^{8ik\ell} = 0\,. 
$$
Similarly to the previous example we can find the positions of the resolvent
resonances $\lambda = k^2$ with $k = \frac{1}{\ell}(2n+1)\pi$ with
multiplicity 2, $k = \frac{1}{\ell}2n\pi$ with multiplicity 3 and $k = \frac{%
1}{\ell}[(2n+1)\pi-i\ln{3}]$ with multiplicity 3, $n\in \hbox{\bb Z}$.

This example shows that the symmetry of the graph does not assure
that it will have the effective size smaller than the bound in 
Theorem~\ref{thm-main1}. This graph has four zeros as eigenvalues of $Q\tilde{\Sigma}$,
hence the effective size is $4\ell $, as we expect from four balanced
vertices. Although this graph is very symmetric, we do not have smaller
effective size in contrary to the example in Theorem 7.3 in \cite{DEL}.

\section*{Acknowledgements}

This research was supported by Grant 15-14180Y of the Czech Science
Foundation. The author thanks colleagues from the Nuclear Physics Institute,
Czech Academy of Sciences and the Faculty of Science,
University of Hradec Kr\'alov\'e, namely P.~Exner and J.~K\v{r}\'{i}\v{z} who
provided insight and expertise as well as referees for the useful remarks.


\section*{References}

\begin{thebibliography}{E{\v{S}}94}
\bibitem[AGA08]{AGA} \textrm{Exner P, Keating J~P, Kuchment P, Sunada T and
Teplyaev A (eds) 2007 \emph{\ Analysis on Graphs and Applications,
Proceedings of a Isaac Newton Institute programme\/} (\emph{Proceedings of
Symposia in Pure Mathematics\/} vol~77) (Providence, R.I.) }

\bibitem[BHJ12]{BHJ} \textrm{Band R, Harrison J~M and Joyner C~H 2012 Finite
pseudo orbit expansions for spectral quantities of quantum graphs \emph{J.
Phys. A: Math. Theor.\/} \textbf{\ 45} 325204 }

\bibitem[BK13]{BK} \textrm{Berkolaiko G and Kuchment P 2013 \emph{%
Introduction to Quantum Graphs\/} Mathematical Surveys and Monographs 186
(AMS) }

\bibitem[BSS10]{BSS} \textrm{Band R, Sawicki A and Smilansky U 2010
Scattering from isospectral quantum graphs \emph{J. Phys. A: Math. Theor.\/} 
\textbf{\ 43} 415201 }

\bibitem[DEL10]{DEL} \textrm{Davies E~B, Exner P and Lipovsk\'{y} J 2010 Non-%
{W}eyl asymptotics for quantum graphs with general coupling conditions \emph{%
J. Phys. A: Math. Theor.\/} \textbf{43} 474013 }

\bibitem[DP11]{DP} \textrm{Davies E~B and Pushnitski A 2011 Non-{W}eyl
resonance asymptotics for quantum graphs \emph{Analysis and PDE\/} \textbf{4}
729--756 }

\bibitem[EL07]{EL1} \textrm{Exner P and Lipovsk\'{y} J 2007 Equivalence of
resolvent and scattering resonances on quantum graphs \emph{Adventures in
Mathematical Physics (Proceedings, Cergy-Pontoise 2006)\/} vol 447
(Providence, R.I.) pp 73--81 }

\bibitem[EL10]{EL2} \textrm{Exner P and Lipovsk\'{y} J 2010 Resonances from
perturbations of quantum graphs with rationally related edges \emph{J. Phys.
A: Math. Theor.\/} \textbf{43} 1053 }

\bibitem[EL11]{EL3} \textrm{Exner P and Lipovsk\'{y} J 2011 Non-{W}eyl
resonance asymptotics for quantum graphs in a magnetic field \emph{Phys.
Lett. A\/} \textbf{375} 805--807 }

\bibitem[E{\v{S}}94]{ESr} \textrm{Exner P and {\v{S}}ere\v{s}ov\'{a} E 1994
Appendix resonances on a simple graph \emph{J. Phys. A\/} \textbf{27}
8269--8278 }

\bibitem[Exn97]{Ex3} \textrm{Exner P 1997 Magnetoresonances on a lasso graph 
\emph{Found. Phys.\/} \textbf{27} 171--190 }

\bibitem[Exn13]{Ex2} \textrm{Exner P 2013 Solvable models of resonances and
decays \emph{Mathematical Physics, Spectral Theory and Stochastic Analysis\/}
ed M~Demuth W~K (Birkh\"{a}user) pp 165--227 }

\bibitem[Har00]{Ha} \textrm{Harmer M 2000 Hermitian symplectic geometry and
extension theory \emph{J. Phys. A: Math. Gen.\/} \textbf{33} 9193--9203 }

\bibitem[KS99]{KoS2} \textrm{Kottos T and Smilansky U 1999 Periodic orbit
theory and spectral statistics for quantum graphs \emph{Ann. Phys., NY\/} 
\textbf{274} 76--124 }

\bibitem[KS00]{KS2} \textrm{Kostrykin V and Schrader R 2000 Kirchhoff's rule
for quantum wires. {II}: {T}he inverse problem with possible applications to
quantum computers \emph{Fortschritte der Physik\/} \textbf{48} 703--716 }

\bibitem[KS04]{KSc} \textrm{Kottos T and Schanz H 2004 Statistical
properties of resonance width for open quantum systems \emph{Waves in Random
Media\/} \textbf{14} S91--S105 }

\bibitem[Kuc08]{Ku3} \textrm{Kuchment P 2008 Quantum graphs: an introduction
and a brief survey \emph{\ Analysis on Graphs and its Applications\/} Proc.
Symp. Pure. Math. (AMS) pp 291--314 }

\bibitem[Lip15]{Li} \textrm{Lipovsk\'y J 2015 Pseudo-orbit expansion for the
resonance condition on quantum graphs and the resonance asymptotics, \emph{%
Acta Physica Polonica A\/} \textbf{128}, 968---973 }
\end{thebibliography}
\end{document}